\documentclass[12pt]{amsart}
\usepackage{amsmath}
\usepackage{mathtools,xcolor,amssymb,combelow,comment}

\newtheorem{theorem}{Theorem}[section]
\newtheorem{lemma}[theorem]{Lemma}

\newtheorem{proposition}[theorem]{Proposition}

\newtheorem{remark}[theorem]{Remark}

\theoremstyle{definition}

\newcommand{\field}[1]{\mathbb{#1}}
\newcommand{\bR}{\field{R}}
\newcommand{\bN}{\field{N}}

\newtheorem*{theorem*}{Theorem}
\newtheorem*{corollary*}{Corollary}
\newcommand{\beqa}{\begin{eqnarray*}}
	\newcommand{\eeqa}{\end{eqnarray*}}

\def\la{\lambda}

\def\cS{\mathcal{S}}

\def\cI{\mathcal{I}}

\def\rd{\bR^d}

\def\rdd{{\mathbb{R}^{2d}}}

\def\R{\right)}

\def\<{\left<}
\def\>{\right>}

\def\mv1{M_v^1}

\def\phas{(x,\omega )}

\def\o{\omega}

\def\R{\mathbb{R}}
\def\Ren{\mathbb{R}^d}
\def\Renn{\rdd}

\def\sch{\mathcal{S}}

\def\Fur{\mathcal{F}}

\def\Sn2{S_{2}(L^{2}(\Ren))}
\def\S1{S_{1}(L^{2}(\Ren))}
\def\sig00{\sigma_{0,0}}

\def\la{\langle}
\def\ra{\rangle}

\newcommand{\GLL}{\mathrm{GL}\left(2d,\mathbb{R}\right)}

\begin{document}
\begin{abstract}
We study the convergence in $L^2$ of the time slicing approximation of Feynman path integrals under low regularity assumptions on the potential. Inspired by the custom in Physics and Chemistry, the approximate propagators considered here arise from a series expansion of the action. The results are ultimately based on function spaces, tools and strategies which are typical of Harmonic and Time-frequency analysis. 
\end{abstract}

\title[Approximation of Feynman path integrals]{Approximation of Feynman path integrals with non-smooth potentials}

\author{Fabio Nicola and S. Ivan Trapasso}
\address{Dipartimento di Scienze Matematiche,
Politecnico di Torino, corso Duca degli Abruzzi 24, 10129 Torino,
Italy}
\address{Dipartimento di Scienze Matematiche,
	Politecnico di Torino, corso Duca degli Abruzzi 24, 10129 Torino,
	Italy}
\email{fabio.nicola@polito.it}
\email{salvatore.trapasso@polito.it}
\maketitle

\section{Introduction}
Consider the Schr\"odinger equation
\begin{equation}\label{equazione}
i\hbar \partial_t u=-\frac{1}{2}\hbar^2\Delta u+V(t,x)u
\end{equation}
where $0<\hbar\leq 1$ and the potential $V(t,x)$, $(t,x)\in\R\times \rd$, is a real-valued function. Feynman's groundbreaking intuition consisted in recasting the corresponding propagator as a formal integral on an infinite dimensional space of paths in configuration space (see \cite{feynman,feynman-higgs}). Since then, several approaches have been proposed to make rigorous this insight: we do not attempt to reconstruct the enormous literature on this topic, but we refer to the books \cite{albeverio,deGosson0,mazzucchi,reed,schulman} and the references therein.\par
The starting point is the formula for the propagator for the free particle ($V\equiv0$): 
\[
U(t,s)f(x)= \frac{1}{(2\pi i (t-s) \hbar)^{d/2}} \int_{\rd} e^{\tfrac{i}{\hbar}\frac{|x-y|^2}{2(t-s)}} f(y)\, dy.
\]
One can notice that the phase in this integral coincides with the corresponding classical action $S(t,s,x,y)$, which is defined in general as follows. First of all, one introduces the action along a path $\gamma$ in $\rd$ by the formula
\[
S[\gamma]=\int_s^t \mathcal{L}(\gamma(\tau),\dot{\gamma}(\tau),\tau)\,d\tau,
\]
$\mathcal{L}$ being the Lagrangian of the corresponding classical system.
Suppose now that for $t-s$ small enough there is only one classical path $\gamma$ (i.e.\ a path satisfying the Euler-Lagrange equations) such that the boundary condition $\gamma(s)=y$, $\gamma(t)=x$ hold. One then defines the action $S(t,s,x,y)$ along that path (alias generating function) as
\begin{equation}\label{azione}
S(t,s,x,y)=\int_s^t \mathcal{L}(\gamma(\tau),\dot{\gamma}(\tau),\tau)\,d\tau.
\end{equation}
When $\mathcal{L}(x,v)=|v|^2/2$ we retrieve $S(t,s,x,y)=|x-y|^2/(2(t-s))$. \par
For a wide class of {\it smooth} potentials with at most quadratic growth at infinity  it was shown in \cite{fujiwara1,fujiwara2} (see also \cite{fujiwara4,fujiwara5,fujiwara6,kumanogo0,kumanogo1,kumanogo2,kumanogo3,kumanogo4,kumanogo5,kumanogo6,yajima}) that the propagator
$U(t,s)$ for $t-s\not=0$ small enough can be represented as an oscillatory integral operator (OIO), namely
\begin{equation}\label{uzero}
U(t,s)f(x)=\frac{1}{(2\pi i (t-s) \hbar)^{d/2}} \int_{\rd} e^{\tfrac{i}{\hbar}S(t,s,x,y)}b(\hbar,t,s,x,y) f(y)\, dy
\end{equation}
for a suitable amplitude function $b(\hbar,t,s,x,y)$. \par
In concrete situations, except for a few cases, there is no hope to obtain the exact propagator in an explicit, closed form. Therefore, it is a common practice in Physics to consider approximate propagators (parametrices)
 $E^{(N)}(t,s)$ defined by
\begin{equation}\label{ezero}
E^{(N)}(t,s)f(x)=\frac{1}{(2\pi i (t-s) \hbar)^{d/2}} \int_{\rd} e^{\tfrac{i}{\hbar}S^{(N)}(t,s,x,y)} f(y)\, dy,
\end{equation}
where \begin{equation}\label{sn}
S^{(N)}(t,s,x,y)=\frac{|x-y|^2}{2(t-s)}+\sum_{k=1}^N W_k(x,y)(t-s)^k
\end{equation}
is essentially a modified $N$-th order Taylor expansion of the action $S$ at $t=s$; see Section 2 below for the precise construction of the functions $W_k$. \par 
It is clear that the operators $E^{(N)}(t,s)$ do no longer satisfy the evolution property $U(t,s)=U(t,\tau) U(\tau,s)$. The spirit of the time slicing approach can be condensed as follows (see the monograph \cite{fujiwarabook} for a comprehensive treatment instead): for any subdivision $\Omega:s=t_0<t_1<\ldots<t_L=t$ of the interval $[s,t]$, consider the composition
\begin{equation}\label{zero0}
E^{(N)}(\Omega,t,s)=E^{(N)}(t,t_{L-1}) E^{(N)}(t_{L-1},t_{L-2})\ldots E^{(N)}(t_1,s),
\end{equation}
which has integral kernel
\begin{multline}\label{zero}
K^{(N)}(\Omega,t,s,x,y)
\\
=\prod_{j=1}^L \frac{1}{(2\pi i(t_j-t_{j-1})\hbar)^{d/2}}\int_{\R^{d(L-1)}} \exp\Big(\tfrac{i}{\hbar}\sum_{j=1}^L S^{(N)}(t_j,t_{j-1},x_j,x_{j-1})\Big) \prod_{j=1}^{L-1} dx_j,
\end{multline}
with $x=x_{L}$ and $y=x_{0}$.\par
It is reasonable to believe that the operators $E^{(N)}(\Omega,t,s)$ converge to the actual propagator as $\omega(\Omega):=\sup\{t_j-t_{j-1}: j=1,\ldots,L\}\to 0$, in line with Feynman's insight. In order to keep the technicalities at minimum, in this note we confine our investigation to the space of bounded operators in $L^2(\rd)$, endowed with the usual operator norm. Nevertheless, this basic framework leaves room for considering potentials characterized by mild regularity assumptions. A suitable reservoir of such functions is the so-called Sj\"ostrand's class, which can be provisionally defined as the space of tempered distributions $\sigma\in\cS'(\rd)$ such that
$$ \int_{\rd} \sup_{x\in\rd} |\langle \sigma, M_{\omega}T_x g\rangle | d \omega<\infty,
$$ for any non-zero Schwartz function $g\in\cS(\rd)\setminus \{0\}$, where $M_{\omega}$ and $T_{x}$, $x,\omega \in \rd$ respectively denote the modulation and translation operators:
$$M_{\omega}f(t)=e^{2\pi i \omega\cdot t}f(t), \qquad T_x f(t) = f(t-x).$$
In particular, the following condition precisely defines the functions we are interested in. \par\medskip
{\bf Assumption (A)} 
{\it $V(t,x)$ is a real-valued function of $(t,x)\in\R\times \rd$ and there exists $N\in \bN$, $N\ge 1$, such that\footnote{We denote by $C^0_b(\R,X)$ the space of continuous and bounded functions $\R\to X$.}
\begin{equation}\label{assumV}
\partial_t^k \partial_x^\alpha V \in C^0_b(\bR,M^{\infty,1}(\rd)),
\end{equation}
for any $k\in\bN$ and $\alpha\in \bN^d$ satisfying $$2k+|\alpha|\leq 2N.$$}
As a rule of thumb, a function in $M^{\infty,1}(\rd)$ is bounded on $\rd$ and locally enjoys the mild regularity of the Fourier transform of an $L^1$ function. This family of functions is largely known in the context of both pseudodifferential operators and phase space analysis - see for instance \cite{grosj} and the references therein.  Although it may seem an exotic symbols class at a first glance, the connections with other function spaces are manifold: for instance, it contains smooth functions all of whose derivatives are bounded and, more generally, any function whose Fourier transform is a finite complex measure. It is worth mentioning that the latter class of potentials has been investigated in relation to path integrals by Albeverio \cite{albeverio} and It\^{o} \cite{ito}. 
Actually, the Sj\"ostrand's class reveals to coincide with a special type of modulation space. In general terms, modulation spaces are Banach spaces defined by controlling the time-frequency concentration and the decay of its elements - see the subsequent section for the details. They were introduced by Feichtinger in the '80s (cf. the pioneering papers \cite{Segal81.Feichtinger_1981_Banach,feichtinger1983modulation}) and they were soon recognized as the optimal environment for Time-frequency Analysis. They also provide a fruitful context to set problems in Harmonic Analysis and PDEs. In particular, several studies on the Schr\"odinger's equation have been conducted from this perspective insofar: among others we mention \cite{CGNR,CNR rough,CN rough} and the references therein. 

\par\smallskip
The aforementioned results motivate in a natural way the problems we consider in this note. We now state our main result. 
\begin{theorem}\label{mainteo}
Assume the condition in Assumption {\rm (A)}. 
For every $T>0$, there exists a constant $C=C(T)>0$ such that, for $0<t-s\leq T\hbar $, $0 < \hbar \le 1$, and any subdivision $\Omega$ of the interval $[s,t]$, we have 
\begin{equation}\label{stimadapro}
\|E^{(N)}(\Omega,t,s)-U(t,s)\|_{L^2\to L^2}
\leq C\omega(\Omega)^{N}.
\end{equation}
\end{theorem}

Notice that estimates of this type were already obtained in \cite{fujiwara2,yajima} for a different type of short-time approximate propagators -- the so-called Birkhoff-Maslov parametices (see also \cite{nicola1,nicola2}). Such parametrices have the form of oscillatory integral operators as in \eqref{uzero}, for suitable amplitudes $b(\hbar,t,s,x,y)$ constructed from the Hamiltonian flow, and satisfy better estimates, with {\it positive} powers of $\hbar$ in the right-hand side of \eqref{stimadapro}. In fact, they are very good approximations of the true propagator both for short-time and $\hbar$ small. However, as already observed, the computation of the exact action $S$ is in general a non trivial task and the study of those parametrices require fine arguments and tools from microlocal analysis. On the other hand, the rougher parametrices in \eqref{ezero} are often preferred by the Physics and Chemistry communities for practical purposes; see \cite{deGosson0,makri1,makri2,makri3}. \par

The paper is organized as follows. In Section 2 we fix the notation and list the preliminary definitions and results. Section 3 contains the construction of a suitable short-time approximation for the action and the proof of some of its properties. In Section 4 we study the operator $E^{(N)}(t,s)$ as a parametrix. In Section 5 we present a general strategy for suitable higher-order parametrices and we prove our main result (Theorem \ref{mainteo}). 

\section{Preliminaries}

\textbf{Notation.} We set $\bN=\left\{0,1,2,\ldots \right\}$ and employ the multi-index notation: in particular, given $\alpha= \left(\alpha_1,\ldots,\alpha_d\right) \in \bN^{d}$ and $x\in\rd$, we write
\[
\left| \alpha \right|=\alpha_1+\ldots+\alpha_d, \qquad  x^{\alpha}=x_1^{\alpha_1}\cdots x_d^{\alpha_d}, 
\]
\[
\partial_x^{\alpha}=\frac{\partial^{\alpha_1}}{\partial x_1^{\alpha_1}}\cdots \frac{\partial^{\alpha_d}}{\partial x_d^{\alpha_d}}.
\]
We set $t^2=t\cdot t$, for $t\in\rd$, and
$xy=x\cdot y$ for the scalar product on $\Ren$. The Schwartz class is denoted by  $\sch(\Ren)$, the space of temperate distributions by  $\sch'(\Ren)$.   The brackets  $\la f,g\ra$ denote the extension to $\sch' (\Ren)\times\sch (\Ren)$ of the inner product $\la f,g\ra=\int f(t){\overline {g(t)}}dt$ on $L^2(\Ren)$. The conjugate exponent $p'$ of $p \in [1,\infty]$ is defined by $1/p+1/p'=1$.

The notation $f\lesssim g$ stands for $f \leq C g$, for a suitable constant $C>0$. 

The Fourier transform of a function $f$ on $\rd$ is normalized as \[
\Fur f(\omega)= \int_{\rd} e^{-2\pi i x\omega} f(x)\, dx,\qquad \omega \in \rd.
\]
Recall the definition of translation and modulation operators: for
any $x,\omega \in\mathbb{R}^{d}$ and $f\in\mathcal{S}\left(\mathbb{R}^{d}\right)$,
\[
\left(T_{x}f\right)\left(t\right)\coloneqq f\left(t-x\right),\qquad\left(M_{\omega}f\right)\left(t\right)\coloneqq e^{2 \pi i \omega t}f\left(t\right).
\]
These operators can be extended by duality on temperate distributions:
for any $x,\omega\in\mathbb{R}^{d}$, $u\in\mathcal{S}'\left(\mathbb{R}^{d}\right)$
and $\varphi\in\mathcal{S}\left(\mathbb{R}^{d}\right)$, we have
\[
\left\langle T_{x}u,\varphi\right\rangle \coloneqq\left\langle u,T_{-x}\varphi\right\rangle ,\qquad\left\langle M_{\omega}u,\varphi\right\rangle \coloneqq\left\langle u,M_{-\omega}\varphi\right\rangle .
\]
In according with a harmless improper custom, we occasionally write
$\left(T_{x}u\right)\left(t\right)=u\left(t-x\right)$ even for $u\in\mathcal{S}'(\mathbb{R}^{d})$.
The composition $\pi\phas=M_\omega T_x$ constitutes a so-called time-frequency shift.

For any $\lambda\neq0$, consider the unitary and non-normalized scalar dilation operators 
on $L^{2}\left(\mathbb{R}^{d}\right)$ defined by 
\begin{equation}\label{scalardilat}
U_{\lambda}f\left(x\right)\coloneqq\left|\lambda\right|^{d/2}f\left(\lambda x\right),\qquad D_{\lambda}f\left(x\right)\coloneqq f\left(\lambda x\right),\qquad f\in L^{2}\left(\mathbb{R}^{d}\right).
\end{equation}
These definitions naturally extend to the case of an  invertible matrix $A\in\mathrm{GL}(d,\bR)$ as
\begin{equation}\label{matrixdilat}
U_{A}f(x)\coloneqq \left| \det A \right|^{1/2}f\left(A x\right), \quad D_{A}f\left(x\right)\coloneqq f\left(A x\right).\end{equation}

\subsection{Short-time Fourier transform} 
The short-time Fourier transform (STFT) of a signal $f\in\cS'(\rd)$ with respect to the window function $g \in \cS(\rd)\setminus\{0\}$ is defined as
\begin{equation}\label{STFTdef}
V_gf\phas=\langle f,\pi\phas g\rangle=\Fur (f\cdot T_x g)(\omega)=\int_{\Ren}
f(y)\, {\overline {g(y-x)}} \, e^{-2\pi iy \o }\,dy.
\end{equation}

For a thorough account on the properties of the STFT see \cite{GRO}. It is important to remark that the STFT is deeply connected with other well-known phase-space transforms, such as the ambiguity distribution
$$ Amb\left(f,g\right)\left(x,\omega\right)=\int_{\rd} e^{-2\pi i\omega y}f\left(y+\frac{x}{2}\right)\overline{g\left(y-\frac{x}{2}\right)}dy, $$ and the Wigner transform
$$W(f,g)(x,\omega )=\int_{\mathbb{R}^{d}}e^{-2\pi iy \omega }f\left(x+\frac{y}{2}\right)%
\overline{g\left(x-\frac{y}{2}\right)}\, dy.$$

In particular, the following relations hold for any $x,\omega \in \rd$ and $f\in \cS'(\rd), g\in \cS (\rd) \setminus \{0\}$: \[A(f,g)(x,\omega)=e^{\pi i x \omega}V_g f (x,\omega), \quad W(f,g)(x,\omega)=2^d e^{4\pi i x \omega} V_{\cI g} f (2x,2\omega),\] where $\cI g (t) = g(-t)$. 

For this and other aspects of the connection with phase space analysis, see also \cite{deGosson1}. 

\subsection{Modulation spaces} 

Given a non-zero window $g\in\sch(\Ren)$ and $1\leq p,q\leq
\infty$, the {\it
	modulation space} $M^{p,q}(\Ren)$ consists of all tempered
distributions $f\in\sch'(\Ren)$ such that $V_gf\in L^{p,q}(\Renn )$
(mixed-norm space). Equivalently, $M^{p,q}$ contains the distributions $f$ such that
$$
\|f\|_{M^{p,q}}=\|V_gf\|_{L^{p,q}}=\left(\int_{\Ren}
\left(\int_{\Ren}|V_gf(x,\o)|^p\,
dx\right)^{q/p}d\o\right)^{1/q}  \, <\infty ,
$$
with trivial adjustments if $p$ or $q$ is $\infty$. 
If $p=q$, we write $M^p$ instead of $M^{p,p}$. 

It can be proved that $M^{p,q}(\rd)$ is a Banach space whose definition does not depend on the choice of the window $g$. For this and other properties we address the reader to \cite{GRO}. 

The Sj\"ostrand's class, originally defined in \cite{sjo}, coincides with the choice $p=\infty$, $q=1$. We have that $M^{\infty,1}(\rd)\subset L^{\infty}(\rd)$ and it is a Banach algebra under pointwise product. In fact, much more is true (cf. \cite{narimani}): for any $1\leq p,q \leq \infty$, the following continuous embedding holds:
$$M^{\infty,1}\cdot M^{p,q} \hookrightarrow M^{p,q}.$$ 

Within the broad family of (weighted) modulation spaces, we can retrieve a number of well-known classical spaces, such as Sobolev or Bessel potential spaces - see again \cite{GRO}. Here, we confine ourselves to remark that $L^2 = M^{2}$. 

For the benefit of the reader, let us recall a result on the boundedness of dilation operators on modulation spaces that will be repeatedly used hereafter - cf. \cite[Theorem 3.1]{sugimoto} and \cite[Proposition 3.1]{dilatCN}. 

\begin{lemma}\label{dilatmod}
Let $1\le p,q \le \infty$ and $A\in \mathrm{GL}(d,\bR)$. For any $f\in M^{p,q}(\rd)$, 
\[ \left\Vert D_{A}f \right\Vert_{M^{p,q}} \lesssim C_{p,q}(A) \left\Vert f \right\Vert_{M^{p,q}}, 
\] where 
\[
C_{p,q}(A) = \left|\det A\right|^{-(1/p-1/q+1)}\left( \det (I+A^{\top}A)\right)^{1/2}. 
\]
\end{lemma}

\section{Short-time action and related estimates}
We begin with a brief discussion devoted to explain the structure of formula \eqref{sn} for the approximate action $S^{(N)}(t,s,x,y)$ appearing in \eqref{ezero}. We refer to \cite[Section 4.5]{deGosson0} for more details.\par
It is well known that, for a classical Hamiltonian of physical type 
\[
H(x,p,t)=\frac{1}{2}p^2 + V(t,x),
\]
 the action $S(t,s,x,y)$ satisfies the Hamilton-Jacobi equation
\begin{equation}\label{HJ}
\frac{\partial S}{\partial t}+\frac{1}{2}|\nabla_x S|^2+V(t,x)=0. 
\end{equation}
In order for $E^{(N)}$ to be a parametrix in a sense to be specified (cf. Remark \ref{discGN} below), we consider the slightly modified equation 
\begin{equation}\label{HJmod}
\frac{\partial S}{\partial t}+\frac{1}{2}|\nabla_x S|^2+V(t,x)+\frac{i\hbar d}{2(t-s)} - \frac{i\hbar}{2}\Delta_x S=0, 
\end{equation} and look for a solution $S$ in the form $S(t,s,x,y)=\frac{|x-y|^2}{2(t-s)}+R(t,s,x,y)$, $s < t$. This yields an equivalent equation for $R$, namely  
\[
\frac{\partial R}{\partial t}+\frac{1}{2}|\nabla_x R|^2+V(t,x)+\frac{1}{t-s}(x-y)\cdot\nabla_x R - \frac{i\hbar}{2}\Delta_x R=0. 
\]
Assume that \[
R(t,s,x,y)=W_0+W_1(s,x,y)(t-s)+W_2(s,x,y)(t-s)^2+\ldots,
\] where the functions $W_k(s,x,y)$ will be briefly denoted by $W_k(x,y)$ from now on. We immediately find $W_0=0$ and, for $k\geq1$, by equating to $0$ the coefficient of the term $(t-s)^{k-1}$, we obtain the equations
\begin{equation}\label{PDEWk}
k W_k(x,y)+(x-y)\cdot \nabla_x W_k(x,y)=F_k(x,y),
\end{equation} where we set
	\begin{equation}\label{Fk} F_k(x,y)=-\frac{1}{2}\sum_{j+\ell=k-1\atop j\geq 1, \ell\geq 1}\nabla_x W_j\cdot \nabla _x W_{\ell}-\frac{1}{(k-1)!}\partial^{k-1}_t V(s,x)+\frac{i\hbar}{2}\Delta_x W_{k-1}.\end{equation}

\begin{lemma}
	For any integer $k\geq 1$ there exists a unique continuous solution of \eqref{PDEWk}, namely
	
	\begin{equation}\label{Wk}
	W_k(x,y)=\int_0^1 \tau^{k-1} F_k(\tau x + (1-\tau)y,y) d\tau .
	\end{equation}
\end{lemma}

\begin{proof}
	According to the methods of characteristics, along the curves of type $x_u(\lambda)=y+ue^\lambda$, where $\lambda\in\mathbb{R}$ and $u\in\rd$ has unitary norm, the original PDE \eqref{PDEWk} becomes a linear ODE with respect to the variable $\lambda$:
	$$ \frac{\mathrm{d}}{\mathrm{d}\lambda}W_k(x_u(\lambda),y) + k W_k (x_u(\lambda),y) = F_k(x_u(\lambda),y),$$ whose solutions are given by 
	$$W_k(x_u(\lambda),y)=e^{-k\lambda}\left(\int_{-\infty}^{\lambda} e^{k\sigma}F_k(x_u(\sigma),y)d\sigma + C \right),$$ where $C\in\mathbb{R}$ is an arbitrary constant. Notice that $\lambda = \log{\|x-y\|}$ and the change of variable $\sigma=\log{( \|x-y\|\tau)}$ thus gives 
	$$	W_k(x,y)=\int_0^1 \tau^{k-1} F_k(\tau x + (1-\tau)y,y)d\tau + \frac{C}{\| x-y \|^k}.$$
	It is therefore clear that the unique continuous solution corresponds to $C=0$. 
\end{proof}	
	
\par

For $N=1,2,\ldots$, we now define the approximate generating functions as in \eqref{sn}, namely
\begin{equation}\label{SN}
S^{(N)}(t,s,x,y)=\frac{|x-y|^2}{2(t-s)}+R^{(N)}(t,s,x,y),
\end{equation}
where 
\begin{equation}\label{RN}
R^{(N)}(t,s,x,y)\coloneqq \sum_{k=1}^N W_k(x,y)(t-s)^k
\end{equation} and $W_k(x,y)$ is defined in \eqref{Wk}. \par

In particular, the first-order approximation of the action $(N=1)$ is
\[
S^{(1)}(t,s,x,y)=\frac{|x-y|^2}{2(t-s)}-(t-s)\int_0^1{V(s,\tau x + (1-\tau)y)d\tau},
\]

It is worth mentioning that determining the correct short-time approximations to the action functional is an important matter, initially investigated by Makri and Miller in \cite{makri1,makri2,makri3} - see also \cite{deGosson2} and the recent paper \cite{deGosson3} for its relevance to  quantization issues. We remark that the authors consider power series solutions of the Hamilton-Jacobi equation \eqref{HJ} instead of \eqref{HJmod}, but this would provide poor approximating power for $E^{(N)}$ as a parametrix, namely a first order error in $t-s$ regardless of $N$. In fact, the results in \cite{fujiwara2} show that the parametrix in \eqref{ezero} with $S^{(N)}$ replaced by the true action $S$ does not enjoy better estimates than a first order one, even for smooth potentials. \\

Concerning the regularity of the terms $W_k$ in $S^{(N)}$, we can prove the following result. 

\begin{proposition}\label{W sj}
	If $V$ satisfies Assumption $\mathrm{(A)}$, then for any $1\leq k \leq N$ we have 
	\[
	\|\partial_x^{\alpha}W_k\|_{M^{\infty,1}(\rdd)}\leq C, \quad \text{for}\ |\alpha|\leq 2(N-k+1),\ s\in\R,
	\]  
	for some constant $C>0$.
\end{proposition}

\begin{proof}
Let us first prove the claim for $k=1$. For $\left(z,\zeta\right)\in\mathbb{R}^{2d}$, the STFT of $\partial_x^{\alpha}W_{1}$, $|\alpha|\leq 2N$, can be written as
	\begin{align*}
	\left| V_{g} \partial_x^{\alpha}W_{1}\left(z,\zeta\right)\right| & =\left|\int_{0}^{1} \tau^{|\alpha|} V_{g}\left[\partial_x^{\alpha}V\left(s,\tau x+\left(1-\tau\right)y\right)\right]\left(z,\zeta\right)d\tau\right|.
	\end{align*}
We now think of $V$ as a function on $\rdd$. More precisely, define
$$V'(s,x,y) \coloneqq V(s,x), \quad s\in\R,\ x,y\in\rd$$ and notice that $V'$ still satisfies Assumption $\mathrm{(A)}$ with $M^{\infty,1}(\rd)$ replaced by $M^{\infty,1}(\rdd)$. \par
Let us introduce the parametrized matrices $M_{\tau}=\left(\begin{array}{cc}
\tau I & \left(1-\tau\right)I\\
0 & I
\end{array}\right)\in\GLL$, with $\tau\in\left(0,1\right]$. We can thus write
$V\left(s,\tau x+\left(1-\tau\right)y\right) = V'(s, M_{\tau} (x,y)) $, and from Lemma \ref{dilatmod} we have $\partial_x^{\alpha}V'(s,M_\tau (x,y))\in M^{\infty,1}(\rdd)$. Therefore,   
\begin{flalign*}
\left\Vert \partial_x^{\alpha}W_{1} \right\Vert _{M^{\infty,1}\left(\mathbb{R}^{2d}\right)} & \lesssim\int_{0}^{1} \tau^{|\alpha|} \left\Vert \partial_x^{\alpha}V'\left(s,M_{\tau}\cdot\right)\right\Vert _{M^{\infty,1}\left(\mathbb{R}^{2d}\right)}d\tau \\
& \lesssim\left(\int_{0}^{1}\tau^{|\alpha|} C_{\infty,1}(M_{\tau})d\tau\right)\left\Vert \partial_x^{\alpha}V'\right\Vert _{M^{\infty,1}\left(\mathbb{R}^{2d}\right)}<C,
\end{flalign*}
where the last estimate follows from the fact that $$C_{\infty,1}(M_{\tau})=\left( \det (I+M_{\tau}^{\top}M_{\tau})\right)^{1/2}$$ is a continuous function of the parameter $\tau\in [0,1]$.

Assume now that the claim holds for any $W_j$ up to a certain $k \leq N-1$ and consider
	\begin{align*}
\left| V_{g} \partial_x^{\alpha}W_{k+1}\left(z,\zeta\right)\right| & =\left|\int_{0}^{1} \tau^{k+|\alpha|} V_{g}\left[\partial_x^{\alpha}F_{k+1}\left(\tau x+\left(1-\tau\right)y,y\right)\right]\left(z,\zeta\right)d\tau\right|.
\end{align*}

It is easy to deduce from \eqref{Fk} and the hypothesis on $W_1,\ldots,W_k$ that $\partial_x^{\alpha}F_{k+1}(x,y)\in M^{\infty,1}(\rdd)$ whenever $|\alpha|\leq 2(N-k)$. Again from Lemma \ref{dilatmod} we have $\partial_x^{\alpha}F_{k+1}(M_\tau (x,y))\in M^{\infty,1}(\rdd)$, and by the same arguments as before we have  
\begin{flalign*}
\left\Vert \partial_x^{\alpha}W_{k+1} \right\Vert _{M^{\infty,1}\left(\mathbb{R}^{2d}\right)} & \lesssim\int_{0}^{1} \tau^{k+|\alpha|} \left\Vert \partial_x^{\alpha}F_{k+1}\left(M_{\tau}\cdot\right)\right\Vert _{M^{\infty,1}\left(\mathbb{R}^{2d}\right)}d\tau \\
& \lesssim\left(\int_{0}^{1}\tau^{k+|\alpha|} C_{\infty,1}(M_{\tau})d\tau\right)\left\Vert \partial_x^{\alpha}W_{k+1}\right\Vert _{M^{\infty,1}\left(\mathbb{R}^{2d}\right)}<C.
\end{flalign*}

The claim is then proved by induction.  
\end{proof}




\begin{proposition}\label{bound exp}
	If the potential function $V$ satisfies Assumption $\mathrm{\left(A\right)}$, then $e^{\frac{i}{\hbar}R^{(N)}}\in M^{\infty,1}\left(\mathbb{R}^{2d}\right)$, with $R^{(N)}$ as in \eqref{RN}. More precisely, 
	$$ \| e^{\frac{i}{\hbar}R^{(N)}}\|_{M^{\infty,1}} \le C(T),$$ for $0\leq t-s \leq T\hbar$, $0<\hbar\leq 1$. 
\end{proposition}
\begin{proof}
 If $V$ satisfies
	Assumption $\mathrm{\left(A\right)}$, Proposition \ref{W sj} holds
	and $\partial^{\alpha}W_{k}\left(x,y\right)\in M^{\infty,1}\left(\mathbb{R}^{2d}\right)$
	for any $\left|\alpha\right|\leq 2(N-k+1)$. In particular, $W_k \in M^{\infty,1}(\rdd)$ for all $k=1,\ldots,N$ and thus $R^{(N)}\in M^{\infty,1}(\rdd)$. 
	
	Given that $M^{\infty,1}\left(\mathbb{R}^{2d}\right)$ is a Banach
	algebra for pointwise multiplication, it is enough to show the desired estimate for   $e^{\frac{i}{\hbar}\left(t-s\right)^kW_{k}}$, for any $1\leq k \leq N$: 

\begin{alignat*}{1}
\left\Vert e^{\frac{i}{\hbar}\left(t-s\right)^kW_{k}}\right\Vert _{M^{\infty,1}} & = \left\Vert \sum_{n=0}^{\infty}\frac{i^{n}\left(t-s\right)^{kn}\left(W_{k}\right)^{n}}{\hbar^{n}n!}\right\Vert _{M^{\infty,1}} \\
& \le\sum_{n=0}^{\infty}\frac{C^{n-1}\left(t-s\right)^{kn}\left\Vert W_{k}\right\Vert _{M^{\infty,1}}^{n}}{\hbar^{n}n!}\\
& =C^{-1}e^{\frac{C}{\hbar}\left(t-s\right)^{k}\left\Vert W_{k}\right\Vert _{M^{\infty,1}}}\leq C(T), 
\end{alignat*}
for $0\leq t-s \leq T\hbar$, $0<\hbar\le 1$. 
\end{proof}

\section{Short-time approximate propagator} 

Let us first recall that the Cauchy problem for the Schr\"odinger equation with bounded potentials is globally well-posed in $L^2(\rd)$. This is an easy and classic result that can be stated as follows.

\begin{proposition}\label{wellpos L2}
	Assume that $V$ is a real-valued function on $\mathbb{R}\times\mathbb{R}^{d}$
	satisfying $V\in C^{\infty}(\bR,L^{\infty}(\rd))$ and let $s\in\mathbb{R}$.
	Then, the Cauchy problem 
	\[
	\begin{cases}
	i\hbar\partial_{t}u=-\frac{1}{2}\hbar^{2}\Delta u+V\left(t,x\right)u\\
	u\left(s,x\right)=u_{0}\left(x\right)
	\end{cases}
	\]
	is (backward and) forward globally well-posed in $L^{2}(\mathbb{R}^{d})$ and the corresponding propagator $U\left(t,s\right)$ is a unitary operator on $L^2(\rd)$. 
\end{proposition}

Consider the parametrix $E^{(N)}(t,s)$ in \eqref{ezero}. We have the following result.
\begin{proposition}
For every $T>0$ there exists $C=C(T)>0$ such that, for $0<t-s\leq T\hbar$, $0<\hbar\leq1$,  we have
\begin{equation}\label{dg2}
\| E^{(N)}(t,s)\|_{L^2\to L^2}\leq C. 
\end{equation}
Moreover, for $f\in L^2(\rd)$ we have 
\begin{equation}\label{dg3}
\lim_{t \searrow s} E^{(N)}(t,s)f=f
\end{equation}
in $L^2(\rd)$. 
\end{proposition}
\begin{proof}
First, notice that
	\[
	E^{\left(N\right)}\left(t,s\right)f\left(x\right)=\frac{1}{\left(2\pi i\left(t-s\right)\hbar\right)^{d/2}}\int_{\mathbb{R}^{d}}e^{\frac{i}{\hbar}\frac{\left|x-y\right|^{2}}{2\left(t-s\right)}}e^{\frac{i}{\hbar}R^{(N)}(t,s,x,y)}f\left(y\right)dy
	\]
	is an OIO with the free-particle-action as phase
	function and amplitude $a^{(N)}\left(t,s,x,y\right)\coloneqq \exp{\left(\frac{i}{\hbar}R^{(N)}(t,s,x,y)\right)}\in M^{\infty,1}\left(\mathbb{R}^{2d}\right)$ by Proposition \ref{bound exp}.
	There is a number of results concerning the $L^{2}$-boundedness in
	this context, but in order to keep track of the short-time behaviour
	in the estimates a few more steps are needed.

	First, notice that the dilation operators defined in \eqref{scalardilat} allow us to rephrase $E^{\left(N\right)}\left(t,s\right)$ as
	follows:
	\[
	E^{\left(N\right)}\left(t,s\right)=U_{\frac{1}{\sqrt{\hbar(t-s)}}}\tilde{E}^{\left(N\right)}\left(t,s\right)U_{\sqrt{\hbar(t-s)}},
	\]
	where
	\[
	\tilde{E}^{\left(N\right)}\left(t,s\right)f\left(x\right)=\frac{1}{\left(2\pi i\right)^{d/2}}\int_{\mathbb{R}^{d}}e^{\frac{i}{2}\left|x-y\right|^{2}}\tilde{a}^{(N)}(t,s,x,y)f\left(y\right)dy
	\]
	is an OIO whose phase function is free from time and $\hbar$ 
	dependence and the amplitude is
	\[
	\tilde{a}^{(N)}\left(t,s,x,y\right)=e^{\frac{i}{\hbar}\sum_{k=1}^{N}W_{k}\left(\sqrt{\hbar(t-s)}x,\sqrt{\hbar(t-s)}y\right)\left(t-s\right)^{k}}=D_{\sqrt{\hbar(t-s)}}a^{(N)}\left(x,y\right).
	\]
	In particular, $\tilde{a}^{(N)}\in M^{\infty,1}(\rdd)$ by Lemma \ref{dilatmod}, and $$ \left\Vert \tilde{a}^{(N)} \right\Vert_{M^{\infty,1}} \leq C(T)\left\Vert a^{(N)} \right\Vert_{M^{\infty,1}}$$ for $0<\hbar(t-s)\le T$. 
		
	We are then able to prove \eqref{dg2} by means of boundedness results for this kind of operators, such as \cite[Theorem 2.1]{boulk}, and Proposition \ref{bound exp}: 
	\[
	\left\Vert E^{\left(N\right)}\left(t,s\right)\right\Vert _{L^{2}\rightarrow L^{2}}=\left\Vert \tilde{E}^{\left(N\right)}\left(t,s\right)\right\Vert _{L^{2}\rightarrow L^{2}}\lesssim\left\Vert \tilde{a}^{(N)}\right\Vert _{M^{\infty,1}}\lesssim \left\Vert a^{(N)}\right\Vert _{M^{\infty,1}} \le C(T),
	\]
	for $0<t-s\le T\hbar$. \par
	For what concerns strong convergence to the identity as $t \searrow s$,
	consider the operator 
	\[
	H^{\left(N\right)}\left(t,s\right)f\left(x\right)=\frac{1}{\left(2\pi i\left(t-s\right)\hbar\right)^{d/2}}\int_{\mathbb{R}^{d}}e^{\frac{i}{\hbar}\frac{\left|x-y\right|^{2}}{2\left(t-s\right)}}\left(e^{\frac{i}{\hbar}R^{(N)}(t,s,x,y)}-1\right)f\left(y\right)dy
	\]
	and employ again the dilations in order to write
	\[
	H^{\left(N\right)}\left(t,s\right)=U_{\frac{1}{\sqrt{\hbar(t-s)}}}\tilde{H}^{\left(N\right)}\left(t,s\right)U_{\sqrt{\hbar(t-s)}},
	\]
	where 
	\[
	\tilde{H}^{\left(1\right)}\left(t,s\right)f\left(x\right)=\frac{1}{\left(2\pi i\right)^{d/2}}\int_{\mathbb{R}^{d}}e^{\frac{i}{2}\left|x-y\right|^{2}}\tilde{b}^{(N)}(t,s,x,y)f\left(y\right)dy
	\]
	is an OIO with amplitude  $\tilde{b}^{(N)}(t,s,x,y)=\tilde{a}^{(N)}\left(t,s,x,y\right)-1\in M^{\infty,1}\left(\mathbb{R}^{2d}\right)$.
	
	The latter can be expanded as follows:
	\begin{align*}
	\tilde{b}^{\left(N\right)}\left(t,s,x,y\right)&=e^{\frac{i}{\hbar}R^{\left(N\right)}\left(t,s,\sqrt{\hbar(t-s)}x,\sqrt{\hbar(t-s)}y\right)}-1\\
	&=\frac{i}{\hbar}\left(t-s\right)\overline{R^{\left(N\right)}}\left(t,s,\sqrt{\hbar(t-s)}x,\sqrt{\hbar(t-s)}y\right),
	\end{align*}
	where 
	\begin{align*}
&\overline{R^{\left(N\right)}}\left(t,s,x,y\right)\\
&=
\sum_{n=1}^\infty \frac{i^{n-1}}{n!}\Big(\frac{t-s}{\hbar}\Big)^{n-1}\Big(\sum_{k=1}^{N}W_{k}\left(\sqrt{\hbar(t-s)}\,x,\sqrt{\hbar(t-s)}\,y\right)\left(t-s\right)^{k-1}\Big)^n.
\end{align*}	
The algebra property of the Sj\"ostrand's class and Lemma \ref{dilatmod} imply that $\overline{R^{\left(N\right)}}$ belongs to a bounded subset of  $ M^{\infty,1}(\rdd)$ for $0<t-s\le T\hbar$, $0<\hbar\leq1$. It is then clear that $\tilde{b}^{(N)}\rightarrow 0$ in $M^{\infty,1}\left(\mathbb{R}^{2d}\right)$ for $t \searrow s$. Therefore, the OIO with operator $\tilde{b}^{(N)}$ has operator norm converging to 0 as $t \searrow s$, and \eqref{dg3} follows. 
\end{proof}

\begin{remark}\label{discGN} \rm A direct check shows that the $E^{(N)}(t,s)$ is a parametrix, meaning that 
\[
\left(i\hbar\partial_t+\frac{1}{2}\hbar^2\Delta-V(t,x)\right) E^{(N)}(t,s)=G^{(N)}(t,s),
\]
with
\begin{equation}\label{defgn}
G^{(N)}(t,s)f =\frac{1}{(2\pi i (t-s) \hbar)^{d/2}} \int_{\rd} e^{\tfrac{i}{\hbar}S^{(N)}(t,s,x,y)} g_N(\hbar,t,s,x,y) f(y)\, dy,
\end{equation}
where, from the construction of $S^{(N)}$ (see in particular eqs. \eqref{sn}, \eqref{PDEWk} and \eqref{Fk}), the amplitude $g_N$ satisfies
\begin{align*}
g_{N}\left(\hbar,t,s,x,y\right) & =-\frac{\partial S^{\left(N\right)}}{\partial t}-\frac{1}{2}\left|\nabla_{x}S^{\left(N\right)}\right|^{2}-V\left(t,x\right)-\frac{i\hbar d}{2\left(t-s\right)}+\frac{i\hbar}{2}\Delta_{x}S^{\left(N\right)}\\
\begin{split}
& =\sum_{k=N}^{2N}-\frac{1}{2}\left(\sum_{\substack{j+\ell=k\\
		j,\ell\ge1
	}
}\nabla_{x}W_{j}\cdot\nabla_{x}W_{\ell}\left(t-s\right)^{k}\right) \\ & + \frac{i\hbar}{2}\Delta_{x}W_{N}\left(x,y\right)\left(t-s\right)^{N}
\\ &  -\frac{\left(t-s\right)^{N}}{\left(N-1\right)!}\int_{0}^{1}\left(1-\tau\right)^{N-1}\left(\partial_{t}^{N}V \right)\left(\left(1-\tau\right)s+\tau t,x\right)d\tau .
\end{split}  
\end{align*}
Hence, by Assumption $\mathrm{(A)}$ and Proposition \ref{W sj}
\[
\left\Vert g_N\left(\hbar,t,s,\cdot,\cdot \right)\right\Vert_{M^{\infty,1}(\rdd)}\le C\left(t-s\right)^N
\]
for $0<t-s\le T$, with a constant $C=C(T)>0$ independent of $\hbar\in (0,1]$. 

Following the path of the preceding proof, by means of suitable dilations we can conveniently recast  $G^{(N)}(t,s)$ as an OIO with time and $\hbar$ independent phase and amplitude 
\[
\tilde{g}^{(N)}(\hbar,t,s,x,y)= D_{\sqrt{\hbar(t-s)}} \left[e^{\frac{i}{\hbar}R^{(N)}(t,s,x,y)}g^{(N)}(\hbar,t,s,x,y)\right]. 
\]
We have $\tilde{g}^{(N)}\in M^{\infty,1}(\rdd)$ by Lemma \ref{dilatmod} - in fact, more than this is true: 
\[
\left\Vert \tilde{g}^{(N)}(\hbar,t,s,\cdot,\cdot) \right\Vert_{M^{\infty,1}} \le C (t-s)^N
\]
for $0<t-s\le T\hbar$ and a constant $C=C(T)>0$. 
Therefore, $G^{(N)}(t,s)$ extends to a bounded operator on $L^2(\rd)$ (cf. again \cite{boulk}) and arguments similar to those of the preceding proof lead to the estimate

\begin{equation}\label{stimaGN}
\left\Vert G^{\left(N\right)}f\right\Vert _{L^{2}}\le C\left\Vert \tilde{g}^{(N)}\right\Vert _{M^{\infty,1}}\left\Vert f\right\Vert _{L^{2}} \le C(T)\left(t-s\right)^N\left\Vert f\right\Vert _{L^{2}},
\end{equation}
for $0<t-s\le T\hbar$. 

\end{remark}
The preceding discussion is the bedrock of the following result. 
\begin{theorem}\label{stimaN+1}
	For every $T>0$, there exists a constant $C=C(T)>0$ such that
	\begin{equation}\label{stimaN+1eq}
	\Vert E^{(N)}(t,s) - U(t,s) \Vert_{L^2\rightarrow L^2} \leq C\hbar^{-1}(t-s)^{N+1},
	\end{equation}
	whenever $0<t-s\leq T\hbar$. 
\end{theorem}
\begin{proof}
	The propagator $U(t,s)$ clearly satisfies the equation
	\[
	\left(i\hbar\partial_{t}-H\right)U(t,s)f=0
	\]
	for every $f\in L^2(\rd)$, where $H=-\left(\hbar^{2}/2\right)\Delta+V$ is the Hamiltonian operator,
	with $V$ as in Assumption $\mathrm{\left(A\right)}$. 
	On the other hand 
	\[
	\left(i\hbar\partial_{t}-H\right)E^{\left(N\right)}(t,s)f=G^{\left(N\right)}\left(t,s\right)f,
	\]
which can be rephrased in integral form (Duhamel's principle) as
	\[
	E^{\left(N\right)}\left(t,s\right)f=U(t,s)f-i\hbar^{-1}\int_{s}^{t}U(t,\tau)G^{\left(N\right)}(\tau,s)fd\tau.
	\]
	
	Therefore, given $f\in L^{2}\left(\mathbb{R}^{d}\right)$,  by \eqref{stimaGN} we have
	\begin{align*}
	\left\Vert U\left(t,s\right)f-E^{\left(N\right)}\left(t,s\right)f\right\Vert _{L^{2}} & =\left\Vert \hbar^{-1}\int_{s}^{t}U(t,\tau)G^{\left(N\right)}(\tau,s)f d\tau\right\Vert _{L^{2}}\\
	& \leq \hbar^{-1}\int_{s}^{t}\left\Vert U(t,\tau)\right\Vert_{L^2\rightarrow L^2}\left\Vert G^{\left(N\right)}(\tau,s)f\right\Vert _{L^{2}}d\tau\\
	& \le C(T)\hbar^{-1}\int_{s}^{t}\left\Vert f \right\Vert_{L^2}(t-s)^{N}d\tau\\
	& \le C'(T)\hbar^{-1}(t-s)^{N+1}\left\Vert f \right\Vert_{L^2},
	\end{align*}
	for $0<t-s\le T\hbar$.
\end{proof}

\section{An abstract result and proof of the main result (Theorem \ref{mainteo})}
We begin by presenting a convergence result for the approximate propagators in its full generality. In fact, it can be regarded as a generalization of \cite[Lemma 3.2]{fujiwara2}. We also use in the proof some ingenious tricks from that paper.\par
\begin{theorem}\label{absres} Assume that for some $\delta>0$ we have a family of operators $E^{(N)}(t,s)$ for $0<t-s\leq \delta$, and $U(t,s)$, $s,t\in\R$, bounded in $L^2(\rd)$, satisfying the following conditions:\par
$U$ enjoys the evolution property $U(t,\tau)U(\tau,s)=U(t,s)$ for every $s<\tau<t$ and for every $T>0$ there exists a constant $C_0\geq 1$ such that 
 \begin{equation}\label{chiave0}
 \|U(t,s)\|_{L^2\to L^2}\leq C_0\ {\rm for}\ 0<t-s\leq T.
 \end{equation}
Moreover, for some constant $C_1>0$ we have
  \begin{equation}\label{chiave1}
\|E^{(N)}(t,s)-U(t,s)\|_{L^2\to L^2}\leq C_1 (t-s)^{N+1}\ {\rm for}\ t-s\leq \delta.
\end{equation}
For any subdivision $\Omega:s=t_0<t_1<\ldots<t_L=t$ of the interval $[s,t]$, with $\omega(\Omega)=\sup\{t_j-t_{j-1}: j=1,\ldots,L\}<\delta$, consider therefore the composition $E^{(N)}(\Omega,t,s)$ in \eqref{zero0}.\par

Then, for every $T>0$ there exists a constant $C=C(T)>0$ such that 
\begin{equation}\label{daver}
\|E^{(N)}(\Omega,t,s)-U(t,s)\|_{L^2\to L^2}
\leq C \omega(\Omega)^{N}(t-s)\ {\rm for}\ 0<t-s\leq T.
\end{equation}
More precisely, 
$$C=C(T)=C_0^2C_1\exp\Big(C_0C_1\omega(\Omega)^{N}T\Big).$$
\end{theorem}
\begin{proof}

Let
\[
R^{(N)}(t,s):=E^{(N)}(t,s)-U(t,s)
\]
so that by \eqref{chiave1},  we have
\begin{equation}\label{chiave}
\|R^{(N)}(t,s)\|\leq C_1 (t-s)^{N+1} \ {\rm for}\ 0<t-s\leq \delta.
\end{equation}
Hence we can write
\begin{multline}\label{dg1}
E^{(N)}(\Omega,t,s)-U(t,s)\\
=\big(U(t,t_{L-1})+R^{(N)}(t,t_{L-1})\big)\ldots \big(U(t_1,s)+R^{(N)}(t_1,s)\big)-U(t,s).
\end{multline}

One expands the above product and obtains a sum of ordered products of operators, where each product has the following structure: {\it from right to left} we have, say, $q_1$ factors of type $U$, $p_1$ factors of type $R^{(N)}$, $q_2$ factors of type $U$, $p_2$ factors of type $R^{(N)}$, etc., up to $q_k$ factors of type $U$, $p_k$ factors of type $R^{(N)}$, to finish with $q_{k+1}$ factors of type $U$. We can schematically write such a product as 
\[
\underbrace{U\ldots U}_{q_{k+1}}\underbrace{R^{(N)}\ldots R^{(N)}}_{p_{k}}\underbrace{U\ldots U}_{q_{k}}\ldots\ldots \underbrace{R^{(N)}\ldots R^{(N)}}_{p_{1}}\underbrace{U\ldots U}_{q_{1}}.
\]
Here $p_1,\ldots,p_k,q_1,\ldots q_k,q_{k+1}$ are non negative integers whose sum is $L$, with $p_j>0$ and we can of course group together the consecutive factors of type $U$, using the evolution property assumed for $U$. Now, for $0<t-s\leq T$ we estimate the $L^2\to L^2$ norm of the above ordered product using the known estimates for each factor, namely \eqref{chiave0} and \eqref{chiave}. In particular, by using the assumption $C_0 \geq 1$, we get
\begin{align*}
&\leq C_0^{k+1}\prod_{j=1}^k \prod_{i=1}^{p_j} C_1 (t_{J_j+i}-t_{J_{j}+i-1})^{N+1}\\
&\leq C_0\prod_{j=1}^k \prod_{i=1}^{p_j} C_0C_1 (t_{J_j+i}-t_{J_{j}+i-1})^{N+1}\
\end{align*}
where $J_j=p_1+\ldots+p_{j-1}+q_1+\ldots +q_j$ for $j\geq2$ and $J_1=q_1$. \par
The sum over $p_1,\ldots,p_k,q_1,\ldots,q_{k+1}$ of these terms is in turn 
\begin{align*}
&\leq C_0\left\{ \prod_{j=1}^L(1+C_0C_1 (t_{j}-t_{j-1})^{N+1}) -1 \right\}\\
&\leq C_0\left\{\exp\left( \sum_{j=1}^L C_0C_1 (t_{j}-t_{j-1})^{N+1}\right) -1 \right\}\\
&\leq C_0\left\{\exp\left(C_0C_1\omega(\Omega)^{N}(t-s)\right) -1\right\}\\
&\leq C_0^2C_1\omega(\Omega)^{N}(t-s)\exp\left(C_0C_1\omega(\Omega)^{N}(t-s)\right)
\end{align*}
where in the last inequality we used $e^{\tau}-1\leq \tau e^\tau$, for $\tau\geq 0$. 

This gives \eqref{daver} with $C=C(T)$ as in the statement and concludes the proof.
\end{proof}

We can now prove our main result.
\begin{proof}[Proof of Theorem \ref{mainteo}]
	The claim follows at once from Theorem \ref{stimaN+1} and Theorem \ref{absres} applied with $T$ replaced by $T\hbar$, $C_0=1$, $C_1=C\hbar^{-1}$, where $C$ is the constant appearing in \eqref{stimaN+1eq}, 
and using $t-s\le T\hbar$. 
\end{proof}

\end{document}